\documentclass[10pt,A4paper,conference]{IEEEtran}
\usepackage{amsthm}
\usepackage{amsmath}
\usepackage{amsfonts}
\usepackage{graphicx}
\usepackage{latexsym}
\usepackage{amssymb}
\usepackage{stmaryrd}
\usepackage{amscd}
\usepackage[small]{caption}

\newcommand{\sbinom}[2]{\left[ \begin{array}{c} #1 \\ #2 \end{array} \right] }

\newcommand{\field}[1]{\mathbb{#1}}

\newcommand{\sP}{\field{P}}

\newcommand{\sG}{\field{G}}

\DeclareMathAlphabet{\mathbfsl}{OT1}{cmr}{bx}{it}
\newcommand{\uuu}{\kern-1pt\mathbfsl{u}\kern-0.5pt}
\newcommand{\vvv}{\kern-1pt\mathbfsl{v}\kern-0.5pt}

\newcommand{\omegaR}{s}

\newcommand{\myboxplus}{\kern1pt\mbox{\small$\boxplus$}}

\makeatletter \DeclareRobustCommand{\sbinom}{\genfrac[]\z@{}}
\makeatother
\newcommand{\G}[2]{\sbinom{{#1}\kern-1pt}{{#2}\kern-1pt}}
\newcommand{\Gq}[2]{\sbinom{{#1}\kern-0.25pt}{{#2}\kern-0.25pt}}

\newcommand{\Ps}{\smash{{\sP\kern-2.0pt}_q\kern-0.5pt(n)}}
\newcommand{\sPs}{\smash{{\sP\kern-1.5pt}_q(n)}}
\newcommand{\Ptwo}{\smash{{\sP\kern-2.0pt}_2\kern-0.5pt(n)}}
\newcommand{\Ptwom}{\smash{{\sP\kern-2.0pt}_2\kern-0.5pt(m)}}
\newcommand{\Ptwonm}{\smash{{\sP\kern-2.0pt}_2\kern-0.5pt(n+m)}}
\newcommand{\Ptwoa}{\smash{{\sP\kern-2.0pt}_2\kern-0.5pt(1)}}
\newcommand{\Ptwob}{\smash{{\sP\kern-2.0pt}_2\kern-0.5pt(2)}}
\newcommand{\Ptwoc}{\smash{{\sP\kern-2.0pt}_2\kern-0.5pt(3)}}
\newcommand{\Ptwod}{\smash{{\sP\kern-2.0pt}_2\kern-0.5pt(4)}}
\newcommand{\Ptwoe}{\smash{{\sP\kern-2.0pt}_2\kern-0.5pt(5)}}
\newcommand{\Ptwof}{\smash{{\sP\kern-2.0pt}_2\kern-0.5pt(6)}}
\newcommand{\Ptwokm}{\smash{{\sP\kern-2.0pt}_2\kern-0.5pt(2k-1)}}
\newcommand{\Pone}{\smash{{\sP\kern-2.5pt}_2\kern-0.5pt(n{-}1)}}

\newcommand{\Gr}{\smash{{\sG\kern-1.5pt}_q\kern-0.5pt(n,k)}}
\newcommand{\Gi}{\smash{{\sG\kern-1.5pt}_q\kern-0.5pt(n,i)}}
\newcommand{\Gj}{\smash{{\sG\kern-1.5pt}_q\kern-0.5pt(n,j)}}
\newcommand{\Grmk}{\smash{{\sG\kern-1.5pt}_q\kern-0.5pt(n,n-k)}}
\newcommand{\Grdk}{\smash{{\sG\kern-1.5pt}_q\kern-0.5pt(2k,k)}}
\newcommand{\Grekappa}{\smash{{\sG\kern-1.5pt}_q\kern-0.5pt(n,e+1-\kappa)}}
\newcommand{\Grtwoekappa}{\smash{{\sG\kern-1.5pt}_q\kern-0.5pt(n,2e+1-\kappa)}}
\newcommand{\Gremkappa}{\smash{{\sG\kern-1.5pt}_q\kern-0.5pt(n,e-\kappa)}}
\newcommand{\Gn}{\smash{{\sG\kern-1.5pt}_2\kern-0.5pt(n,n{-}1)}}
\newcommand{\Gnq}{\smash{{\sG\kern-1.5pt}_q\kern-0.5pt(n,n{-}1)}}
\newcommand{\Gone}{\smash{{\sG\kern-1.5pt}_2\kern-0.5pt(n,1)}}
\newcommand{\Gqone}{\smash{{\sG\kern-1.5pt}_q\kern-0.5pt(n,1)}}
\newcommand{\GTwo}{\smash{{\sG\kern-1.5pt}_2\kern-0.5pt(n,k)}}
\newcommand{\GTwonk}[2]{{\smash{{\sG\kern-1.5pt}_2\kern-0.5pt({#1},{#2})}}}
\newcommand{\Gnk}{\smash{{\sG\kern-1.5pt}_2\kern-0.5pt(n,n{-}k)}}
\newcommand{\Greone}{\smash{{\sG\kern-1.5pt}_q\kern-0.5pt(n,e{+}1)}}
\newcommand{\Gretwo}{\smash{{\sG\kern-1.5pt}_q\kern-0.5pt(n,e{+}2)}}

\newcommand{\be}[1]{\begin{equation}\label{#1}}
\newcommand{\ee}{\end{equation}}

\newcommand{\Cref}[1]{Co\-rol\-la\-ry\,\ref{#1}}

\newtheorem{theorem}{Theorem}

\newtheorem{remark}{Remark}
\newtheorem{corollary}[theorem]{Corollary}

\newtheorem{construction}{Construction}

\begin{document}

\title{PIR Array Codes with Optimal PIR Rates}

\author{\IEEEauthorblockN{Simon R. Blackburn}
\IEEEauthorblockA{Dept. of Mathematics\\
Royal Holloway University of London\\
Egham, Surrey TW20 0EX, United Kingdom \\
Email: s.blackburn@rhul.ac.uk} \and
\IEEEauthorblockN{Tuvi Etzion}
\IEEEauthorblockA{Dept. of Computer Science\\
Technion-Israel Institute of Technology\\
Haifa 32000, Israel \\
Email: etzion@cs.technion.ac.il}}

\maketitle
\begin{abstract}
There has been much recent interest in Private information Retrieval (PIR) in models
where  a database is stored across several servers using coding techniques
from distributed storage, rather than being simply replicated. In particular,
a recent breakthrough result of Fazelli, Vardy and Yaakobi introduces the notion
of a PIR code and a PIR array code, and uses this notion to produce efficient protocols.

In this paper we are interested in designing PIR array codes.
We consider the case when we have $m$ servers, with each server storing
a fraction $(1/\omegaR)$ of the bits of the database; here $\omegaR$ is a fixed rational
number with $\omegaR > 1$. We study the maximum PIR rate of a PIR array code with
the $k$-PIR property (which enables a $k$-server PIR protocol to be emulated on the $m$ servers),
where the PIR rate is defined to be $k/m$. We present
upper bounds on the achievable rate, some constructions, and ideas how to obtain PIR array
codes with the highest possible PIR rate. In particular, we present constructions that asymptotically
meet our upper bounds, and the exact largest PIR rate is obtained when
$1 < \omegaR \leq 2$.
\end{abstract}

%%%%%%%%%%%%%%%%%%%%%%%%%%%%%%%%%
%%%
%%%     INTRODUCTION
%%%
%%%%%%%%%%%%%%%%%%%%%%%%%%%%%%%%%%%%

\section{Introduction}

A Private Information Retrieval (PIR) protocol allows a user to retrieve
a data item from a database, in such a way that the servers storing the data
will get no information about which data item was retrieved. The problem was introduced in~\cite{CGKS98}.
The protocol to achieve this goal assumes that the servers are curious but honest,
so they don't collude. It is also assumed that the database is error-free
and synchronized all the time. For a set of $k$ servers, the goal is to design a $k$-server PIR protocol,
in which the efficiency of the PIR is measured by the total number of bits transmitted by all parties involved. This model is called a \emph{information-theoretic} PIR; there is also \emph{computational}
PIR, in which the privacy is defined in terms of the inability of a server to compute
which item was retrieved in reasonable time~\cite{KuOs97}. In this paper we will be concerned
only with information-theoretic PIR.

The classical model of PIR assumes that each server stores
a copy of an $n$-bit database, so the \emph{storage overhead}, namely the ratio between the total
number of bits stored by all servers and the size of the
database, is~$k$. However, recent work combines PIR protocols
with techniques from distributed storage (where each server stores only some of the database)
to reduce the storage overhead. This approach was first considered in~\cite{SRR14},
and several papers have developed this direction further:~\cite{ALS14,BEP16,CHY14,FVY15,FVY15a,RaVa16,TaElR16,ZWWG16}.
Our discussion will follow the breakthrough approach presented
by Fazeli, Vardy, and Yaakobi~\cite{FVY15,FVY15a}, which shows that $m$ servers
(for some $m>k$) may emulate a $k$-server PIR protocol with storage overhead significantly lower than $k$.

Fazeli et al~\cite{FVY15a} introduce the key notion of a $[t\times m,p]$ $k$-PIR array code,
which is defined as follows. Let $x_1,x_2,\ldots ,x_p$ be a basis of a vector
space of dimension $p$ (over some finite field $\mathbb{F}$). A \emph{$[t\times m,p]$ array
code} is simply a $t\times m$ array, each entry containing a linear combination of the basis
elements $x_i$. A $[t\times m,p]$ array code satisfies the \emph{$k$-PIR property}
(or is a \emph{$[t\times m,p]$ $k$-PIR array code}) if for every $i\in\{1,2,\ldots ,p\}$
there exist $k$ pairwise disjoint subsets $S_1,S_2,\ldots,S_k$ of columns so
that for all $j\in\{1,2,\ldots ,k\}$ the element $x_i$ is contained in the linear
span of the entries of the columns $S_j$. The following example of a (binary)
$[7\times 4,12]$ $3$-PIR array code is taken from~\cite{FVY15a}:

\begin{scriptsize}
\[
\begin{array}{|c|c|c|c|}\hline
x_1&x_2&x_3&x_1+x_2+x_3\\\hline
x_2&x_3&x_1&x_6\\\hline
x_4&x_5&x_4+x_5+x_6&x_4\\\hline
x_5&x_6&x_8&x_9\\\hline
x_7&x_7+x_8+x_9&x_9&x_7\\\hline
x_8&x_{10}&x_{11}&x_{12}\\\hline
x_{10}+x_{11}+x_{12}&x_{11}&x_{12}&x_{10}.\\\hline
\end{array}
\]
\end{scriptsize}
The $3$-PIR property means that for all $i\in\{1,2,\ldots,12\}$ we can
find $3$ disjoint subsets of columns whose entries span a subspace containing $x_i$.
For example, $x_5$ is in the span of the entries in the subsets
$\{1\}$, $\{2\}$ and $\{3,4\}$ of columns; $x_{11}$ is in the span of the
entries in the subsets $\{1,4\}$, $\{2\}$ and $\{3\}$ of columns.

In the example above, many of the entries in the array consist
of a single basis element; we call such entries \emph{singletons}.

Fazeli et al use a $[t\times m,p]$ $k$-PIR array code as follows. The database
is partitioned into $p$ parts $x_1,x_2,\ldots ,x_p$, each part encoded as an
element of the finite field $\mathbb{F}$. Each of a set of $m$ servers
stores $t$ linear combinations of these parts; the $j$th server stores linear
combinations corresponding to the $j$th column of the array code. We say that
the $j$th server has $t$ \emph{cells}, and stores one linear combination
in each cell. They show that the $k$-PIR property of the array
code allows the servers to emulate all known efficient $k$-server PIR protocols.
But the storage overhead is $tm/p$, and this can be significantly smaller
than $k$ if a good array code is used.
Define $\omegaR=p/t$, so $\omegaR$ can be thought of as the reciprocal of the proportion of the database stored on each server.
For small storage overhead, we would like the ratio
\begin{equation}
\label{eq:ratio_overhead}
\frac{k}{tm/p}=\omegaR\frac{k}{m}
\end{equation}
to be as large as possible. We define the \emph{PIR rate} (\emph{rate} in short)
of a $[t\times m,p]$ $k$-PIR array code to be $k/m$
(this rate should not be confused with the rate of the code). In applications, we would like
the rate to be as large as possible for several reasons: when $\omegaR$, which represents
the amount of storage required at each server, is fixed such schemes give small storage
overhead compared to $k$ (see (\ref{eq:ratio_overhead})); we wish to use a minimal
number $m$ of servers, so $m$ should be as small as possible;
large values of $k$, compared to $m$, are desirable, as they
lead to protocols with lower communication complexity.
We will fix the number $t$ of cells in a server, and the proportion $1/\omegaR$ of the
database stored per server and we seek to maximise the PIR rate.
Hence, we define $g(\omegaR,t)$
to be the largest rate of a $[t\times m,p]$ $k$-PIR array code when $\omegaR$ and $t$
(and so $p$) are fixed. We define $g(\omegaR)=\overline{\lim}_{t\rightarrow\infty}g(\omegaR,t)$.

Most of the analysis in~\cite{FVY15,FVY15a} was restricted to the case $t=1$.
The following two results presented in~\cite{FVY15a} are the most relevant for
our discussion. The first result corresponds to the case where each server holds
a single cell, i.e. we have a PIR code (not an array code with $t>1$).

\begin{theorem}
\label{thm:t=1}
For any given positive integer $\omegaR$, $g(\omegaR,1) = (2^{\omegaR-1})/(2^\omegaR-1)$.
\end{theorem}

The second result is a consequence of the only construction of PIR array codes
given in~\cite{FVY15a} which is not an immediate consequence of the constructions
for PIR codes.

\begin{theorem}
\label{thm:FVYbound}
For any integer $\omegaR \geq 3$, we have $g(\omegaR,\omegaR-1) \geq \omegaR / (2\omegaR-1)$.
\end{theorem}

The goal of this paper is first to generalize the results of Theorems~\ref{thm:t=1}
and~\ref{thm:FVYbound}
and to find codes with better rates for a given $\omegaR$.
We would like to find out the behavior of $g(\omegaR,t)$ as a function of $t$.
This will be done by providing several new constructions for $k$-PIR array codes
which will imply lower bounds on $g(\omegaR,t)$ for a large range of pairs $(\omegaR,t)$.
This will immediately imply a related bound on $g(\omegaR)$ for various values of $\omegaR$.
Contrary to the construction in~\cite{FVY15a}, the value of $\omegaR$ in
our constructions is not necessarily an integer (this possible feature was
mentioned in~\cite{FVY15a}): each rational number greater than one will be considered.
We will also provide various upper bounds on $g(\omegaR,t)$, and related upper bounds on $g(\omegaR)$.
It will be proved that some of the upper bounds on $g(\omegaR,t)$ are tight and also our main upper
bound on $g(\omegaR)$ is tight.

To summarise, our notation used in the remainder of the paper is given by:
\begin{enumerate}
\item $n$ - the number of bits in the database.

\item $p$ - number of parts the database is divided into. The parts will
be denoted by $x_1 , x_2 , \ldots , x_p$.

\item $\frac{1}{\omegaR}$ - the fraction of the database stored on a server.

\item $m$ - the number of servers (i.e. the number of columns in the array).

\item $t$ - number of cells in a server (or the number of rows in the array); so $t=p/\omegaR$.

\item $k$ - the array code allows the servers to emulate a $k$-PIR protocol.

\item $g(\omegaR,t)$ - the largest PIR rate of a $[t\times m,p]$ $k$-PIR array code.

\item $g(\omegaR)=\overline{\lim}_{t\rightarrow\infty}g(\omegaR,t)$.
\end{enumerate}

Clearly, a PIR array code is characterized by the parameters, $\omegaR$, $t$, $k$, and $m$
(the integer $n$ does not have any effect on the other parameters, except for some possible divisibility conditions).
In~\cite{FVY15a}, where the case $t=1$ was considered, the goal was to find the smallest $m$ for given $\omegaR$ and $k$.
This value of $m$ was denoted by the function $M(\omegaR,k)$. The main discussion in~\cite{FVY15a} was
to find bounds on $M(\omegaR,k)$ and to analyse the redundancy $M(\omegaR,k)-s$ and the storage overhead $M(\omegaR,k)/\omegaR$.
When PIR array codes are discussed, the extra parameter is $t$ and given $\omegaR$, $t$, and $k$, the goal
is to find the smallest $m$. We denote this value of $m$ by $M(\omegaR,t,k)$. Clearly, $M(\omegaR,t,k) \leq M(\omegaR,k)$, but
the main target is to find the range for which $M(\omegaR,t,k) < M(\omegaR,k)$, and especially when the storage overhead is low.
Our discussion answers some of these questions, but unfortunately not for small storage overhead
(our storage overhead is much smaller than $k$ as required, but $k$ is relatively large). Hence, our results provide an indication
of the target to be achieved, and this target is left for future work.
We will fix two parameters, $t$ and $\omegaR$,
and examine the ratio $k/m$ (which might require both $k$ and $m$ to be large and as a consequence the storage overhead won't be low).
To have a lower storage overhead we probably need to compromise on a lower ratio of $k/m$.

The rest of this paper is organized as follows. In Section~\ref{sec:upper_bound} we
present a simple upper bound on the value of $g(\omegaR)$. Though this bound is attained, we prove that
$g(\omegaR,t)<g(\omegaR)$ for any fixed values of $\omegaR$ and $t$. We will also state a more complex upper bound
on $g(\omegaR,t)$ for various pairs $(\omegaR,t)$, and it will be shown to be attainable when $1< \omegaR \leq 2$.
In Section~\ref{sec:constructions}
we present a range of explicit constructions.
In Subsection~\ref{sec:s=2minus} we consider the case where $1 < \omegaR \leq 2$.
In Subsection~\ref{sec:rational} we consider the case where $\omegaR$ is rational number greater than~2.
In Section~\ref{sec:Asingletons} we present a construction in which at least $t-1$ cells
in each server are singletons. In Section~\ref{sec:Asingletons} we present a construction in which at least $t-1$ cells
in each server are singletons and its rate asymptotically meets
the upper bound. We believe that this construction always produces the best bounds
and prove this statement in some cases.
For lack of space we omit some proofs and some constructions. These can be found
in the full version of this paper~\cite{BlEt16}.

\section{Upper Bounds on the PIR Rate}
\label{sec:upper_bound}

In this section we will be concerned first with a simple general upper bound (Theorem~\ref{thm:upper_bound})
on the rate of a $k$-PIR array code for a fixed value of $\omegaR$ with $\omegaR>1$. This bound cannot be attained,
but is asymptotically optimal (as $t\rightarrow\infty$). This will motivate us to give
a stronger upper bound  (Theorem~\ref{thm:up_1<s<2}) on the rate $g(\omegaR,t)$ of
a $[t\times m,\omegaR t]$ $k$-PIR array code for various values of $t$ that can sometimes be attained.

\begin{theorem}
\label{thm:upper_bound}
For each rational number $\omegaR > 1$ we have that $g(\omegaR) \leq (\omegaR+1)/(2\omegaR)$. There is no $t$
such that $g(\omegaR,t) = (\omegaR+1)/(2\omegaR)$.
\end{theorem}
\begin{proof}
Suppose we have a $[t\times m,p]$ $k$-PIR array code with $p/t=\omegaR$. To prove the theorem,
it is sufficient to show that $k/m<(\omegaR+1)/(2 \omegaR )$. Since the $k$-PIR property only depends on the span of the contents of a server's cells, we may assume,
without loss of generality, that if $x_i$ can be derived from information
on a certain server then the singleton $x_i$ is stored as the value of one of the cells of this server.

Let $\alpha_i$ be the number of servers which hold the singleton $x_i$ in one of their cells.
Since each server has $t$ cells, we find that $\sum_{i=1}^p \alpha_i\leq tm$,
and so the average value of the integers $\alpha_i$ is at most $tm/p=m/\omegaR$. So there
exists $u\in\{1,2,\ldots p\}$ such that $\alpha_u\leq m/\omegaR$ (and we can only
have $\alpha_u= m/\omegaR$ when $\alpha_i= m/\omegaR$ for all $i\in\{1,2,\ldots ,p\}$).
Let $S^{(1)},S^{(2)},\ldots ,S^{(k)}\subseteq\{1,2,\ldots,m\}$ be disjoint
sets of servers, chosen so the span of the cells in each subset of servers
contains $x_u$. Such subsets exist, by the definition of a $k$-PIR array code.
If no server in a subset $S^{(j)}$ contains the singleton $x_u$, the subset $S^{(j)}$
must contain at least two elements (because of our assumption on singletons stated in the first paragraph of the proof).
So at most $\alpha_u$ of the subsets $S^{(j)}$ are of cardinality $1$.
In particular, this implies that $k\leq \alpha_u+(m-\alpha_u)/2$. Hence
\begin{equation}
\label{eq:u_bound}
\frac{k}{m} \leq \frac{\alpha_u + (m-\alpha_u)/2}{m}
=\frac{1}{2} + \frac{\alpha_u}{2m}
%\leq \frac{1}{2} + \frac{m/\omegaR}{2m}
%= \frac{1}{2} + \frac{1}{2\omegaR} = \frac{\omegaR+1}{2\omegaR}.
\end{equation}
$$
\leq \frac{1}{2} + \frac{m/\omegaR}{2m}
= \frac{1}{2} + \frac{1}{2\omegaR} = \frac{\omegaR+1}{2\omegaR}.
$$
We can only have equality in \eqref{eq:u_bound} when $\alpha_i=m/\omegaR$ for all $i\in\{1,2,\ldots,p\}$,
which implies that all cells in every server are singletons. But then
the span of subset of servers contains $x_i$ if and only if it contains
server with a cell $x_i$, and so $k\leq\alpha_i=m/\omegaR$. But this implies
that the rate $k/m$ of the array code is at most $1/\omegaR=2/(2\omegaR )$. This contradicts
the assumption that the rate of the array code is $k/m=(\omegaR+1)/(2 \omegaR )$,
since $\omegaR>1$. So $k/m<(\omegaR+1)/(2 \omegaR )$, as required.
\end{proof}

\begin{theorem}
\label{thm:up_1<s<2}
For any integer $t \geq 2$ and any positive integer $d$,
we have
\[
g(1 + \frac{d}{t},t) \leq \frac{(2d+1)t +d^2}{(t+d)(2d+1)}=1 - \frac{d^2+d}{(t+d)(2d+1)}.
\]
\end{theorem}

\begin{remark}
We note that we can always write $s=1+d/t$ whenever $s>1$, since $s=p/t$. So Theorem~\label{thm:up_1<s<2} places no extra restrictions on $s$.
\end{remark}

\section{Constructions and Lower Bounds}
\label{sec:constructions}

In this section we will propose various constructions for PIR array codes; these yield
lower bounds on $g(\omegaR,t)$ and on $g(\omegaR)$.
The constructions yield an improvement
on the lower bound on $g(\omegaR)$ implied by Theorem~\ref{thm:FVYbound}. They also cover all rational values of $\omegaR>1$,
and not just integer values of $\omegaR$. We are interested in constructions in which
the number of servers is as small as possible, although the main
goal in this paper is providing a lower bound on the rate.
In the constructions below, we use Hall's marriage Theorem~\cite{Hal35}:

\begin{theorem}
\label{thm:Hall}
In a finite bipartite graph $G=(V_1\cup V_2,E)$,
there is perfect matching if for each subset $X$ of $V_1$, the number
of vertices in $V_2$ connected to vertices of $X$ has at least size $|X|$.
\end{theorem}

\begin{corollary}
\label{cor:Hall}
A finite regular bipartite graph has a perfect matching.
\end{corollary}

\subsection{Constructions for $1 < \omegaR \leq 2$}
\label{sec:s=2minus}

In this subsection we present constructions for PIR array codes when $\omegaR$
is a rational number greater than 1 and smaller than or equal to 2. The first construction will be generalized in Subsection~\ref{sec:rational} and Section~\ref{sec:Asingletons},
when $\omegaR$ is any rational number greater than~$1$, but the special case considered
here deserves separate attention for three reasons: it is simpler than
its generalization; the constructed PIR array code attains the bound
of Theorem~\ref{thm:up_1<s<2}, while we do not have a proof of a similar result
for the generalization; and finally the analysis of the generalization is
slightly different.

\begin{construction}{($\omegaR= 1 + d/t$ and $p= t+d$ for $t >1$, $d$ a positive integer, $1 \leq d \leq t$)}.
\label{con:1<s<2}

Let $\vartheta$ be the least common multiple of $d$ and $t$.
There are two types of servers. Servers of Type A store $t$ singletons.
Each possible $t$-subset of parts occurs $\vartheta/d$ times as the set
of singleton cells of a server, so there are $\binom{p}{t}\vartheta/d$
servers of Type A. Each server of Type B has $t-1$ singleton cells
in $t-1$ cells; the remaining cell stores the sum of the remaining $p-(t-1)=d+1$ parts.
Each possible $(t-1)$-set of singletons occurs $\vartheta/t$ times,
so there are $\binom{p}{t-1}\vartheta/t$ servers of Type B.
\end{construction}

\begin{theorem}
\label{thm:1<s<2}
When $t>1$ and $1 \leq d \leq t$,
\[
g(1+d/t,t) \geq \frac{(2d+1)t + d^2}{(t+d)(2d+1)}.
\]
\end{theorem}
\begin{proof}
The total number of servers in Construction~\ref{con:1<s<2} is
$m =\binom{t+d}{t} \vartheta/d + \binom{t+d}{d+1} \vartheta/t$.
We now calculate $k$ such that Construction~\ref{con:1<s<2} has the $k$-PIR property.
To do this, we compute for each $i$, $1 \leq i \leq p$, a collection
of pairwise disjoint sets of servers, each of which can recover the part~$x_i$.

There are $\binom{t+d-1}{t-1}\vartheta/d$ servers of Type A containing $x_i$ as
a singleton cell. Let $V_1$ be the set of $\binom{t+d-1}{t}\vartheta/d$ remaining
servers of Type A. There are $\binom{t+d-1}{t-2}\vartheta/t$ servers of Type B
containing $x_i$ as a singleton cell. Let $V_2$ be the set of $\binom{t+d-1}{t-1}\vartheta/t$ remaining servers of Type~B.

We define a bipartite graph $G=(V_1 \cup V_2 ,E)$ as follows.
Let $v_1\in V_1$ and $v_2\in V_2$. Let $X_1\subseteq \{x_1,x_2,\ldots,x_p\}$ be
the set of $t$ singleton cells of the server $v_1$. Let $X_2\subseteq \{x_1,x_2,\ldots,x_p\}$
be the parts involved in the non-singleton cell of the server $v_2$.
(So $X_2$ is the set of $d+1$ parts that are not singleton cells of $v_2$. Note that $x_i\in X_2$.)
We draw an edge from $v_1$ to $v_2$ exactly when $X_2\setminus\{x_i\}\subseteq X_1$.
Note that $v_1$ and $v_2$ are joined by an edge
if and only if the servers $v_1$ and $v_2$ can together recover $x_i$.

The degrees of the vertices in $V_1$ are all equal; the same is true for the vertices in $V_2$. Moreover,
$|V_1|=\binom{t+d-1}{t}\vartheta/d =
\binom{t+d-1}{t-1} \vartheta/t=|V_2|$.
So $G$ is a regular graph, and hence by Corollary~\ref{cor:Hall} there exists
a perfect matching in $G$. The  edges of this matching form $|V_1|$ disjoint pairs of servers, each of which can recover $x_i$.
Thus, we have that $k=\binom{t+d-1}{t-1}\vartheta/d + \binom{t+d-1}{t-2} \vartheta/t
+ \binom{t+d-1}{t} \vartheta/d = m - \binom{t+d-1}{t} \vartheta/d$.

Finally, some simple algebraic manipulation shows us that
\[
g(1+d/t,t) \geq \frac{k}{m} = \frac{(2d+1)t + d^2}{(t+d)(2d+1)}~.\qedhere
\]
\end{proof}

\begin{corollary}
\label{cor:1<s<2}
$~$
\begin{enumerate}
\item[\emph{(i)}] For any given $t$ and $d$, $1 \leq d \leq t$, when $\omegaR=1+d/t$ we have
\[
g(\omegaR,t) = 1 - \frac{d^2+d}{(t+d)(2d+1)}
=\frac{\omegaR+1+1/d}{(2+1/d)\omegaR}~.
\]
\item[\emph{(ii)}] For any rational number $1< \omegaR \leq 2$,  we have $g(\omegaR) = (\omegaR +1)/(2 \omegaR)$.
\item[\emph{(iii)}] $g(2,t) = (3t+1)/(4t+2)$.
\end{enumerate}
\end{corollary}

\begin{construction}{($\omegaR= 1 + d/t$, $p= t+d$, and there exists a
Steiner system $S(d,d+1,p)$)}
\label{con:t_t+d}
$~$

Let $\mathcal{S}$ be a $S(d,d+1,p)$ Steiner system on the set of points $\{1,2,\ldots,p\}$.
We define servers of two types. There are $\binom{t+d}{t}=\binom{t+d}{d}$ servers of Type A:
each server stores a different subset of parts in $t$ singleton cells.
There are $\frac{d}{d+1} \binom{t+d}{d}$ servers of Type B, indexed by
a set that repeats each of the $\frac{1}{d+1} \binom{t+d}{d}$ blocks
$B\in\mathcal{S}$ a total of $d$ times. One cell in a server of Type B
contains the sum $\sum_{i\in B}x_i$; the remaining $t-1$ cells contain
the $t-1$ parts not involved in this sum.
\end{construction}

The PIR rate of Construction~\ref{con:t_t+d} attains the upper bound of Theorem~\ref{thm:up_1<s<2}
using fewer servers than in Construction~\ref{con:1<s<2}.
Unfortunately, Construction~\ref{con:t_t+d} can be applied on a limited number of parameters
since the number of possible Steiner systems of this type is limited,
and the number of known ones is even smaller.

\subsection{Constructions when $\omegaR>2$ is rational}
\label{sec:rational}

We do not know the exact value of the asymptotic rate $g(\omegaR,t)$ of PIR codes when $\omegaR>2$.
These values will be considered in this subsection. We present
only the bounds implied by the constructions given in~\cite{BlEt16}.
In all the constructions there exist servers with fewer than $t-1$ singletons.

\begin{theorem}
\label{thm:s=rational}
For given $t$, $d$ and $r$ with $r >1$, with $r \leq t$, and with $1 \leq d \leq t-1$,
$$
g(r + d/t,t) \geq \frac{(rt+d)(rt+d) -t(t-r)}{(rt+d)(2rt+2d-2t+r)} ~.
$$
\end{theorem}
Combining Theorems~\ref{thm:upper_bound} and~\ref{thm:s=rational} we have:
\begin{corollary}
\label{cor:s=rational}
If $\omegaR > 2$ is a rational number which is not an integer, then $g(\omegaR)= (\omegaR+1)/(2 \omegaR)$.
\end{corollary}

\begin{theorem}
\label{thm:s=integer}
For any given integers $\omegaR \geq 2$ and $t \geq \omegaR$,
$$
g(\omegaR,t) \geq \frac{\omegaR t+t+1}{\omegaR (2t+1)}=1 - \frac{(\omegaR -1)(t+1)}{\omegaR (2t+1)} ~.
$$
\end{theorem}
Combining Theorems~\ref{thm:upper_bound} and~\ref{thm:s=integer} we have:
\begin{corollary}
\label{cor:s=integer}
For any given integer $\omegaR > 2$, $g(\omegaR) = \frac{\omegaR +1}{2 \omegaR }$.
\end{corollary}

All the results we obtained are for $t \geq \omegaR -1$. The next theorem can be applied for $t < \omegaR-1$.
\begin{theorem}
If $c$, $\omegaR$, $t$ are integers such that $1 \leq c \leq t-1$ and
$2^{c-1} t - 2^{c-1} (c-2) +1 \leq \omegaR \leq 2^c t - 2^c (c-1)$,
then $g(\omegaR,t) \geq \frac{t-c +(t-1) \omegaR+1}{t-c +2(t-1)\omegaR +2}$.
\end{theorem}

\section{Servers with at least $t-1$ Singletons}
\label{sec:Asingletons}

All the lower bounds described above can be improved with
a construction which generalizes Constrution~\ref{con:1<s<2}.
This general construction can be applied for all admissible pairs $(\omegaR ,t)$.
For simplicity we will define and demonstrate it first for integer values of $\omegaR$ and later
explain the modification needed for non-integer values of~$\omegaR$.

The construction uses $\omegaR$ ($\lceil \omegaR \rceil$ if
not an integer) types of servers.
Type T$_r$, $1 \leq r \leq \omegaR$, has $t-1$ singleton cells and one cell
with a sum of $(r-1)t+1$ parts. For each type, all possible combinations of parts and sums are taken
the same amount of times: $\eta_r$ times for Type T$_r$. Therefore, the number of servers in Type T$_1$ is $\eta_1\binom{\omegaR t}{t}$ and the number of servers in Type T$_r$, $2 \leq r \leq \omegaR$,
is $\eta_r\binom{\omegaR t}{t-1} \binom{\omegaR t-t+1}{(r-1)t+1}$. A part $x_i$ is recovered from all
the singleton cells, where it appears,
and also by pairing servers as follows. We construct $\omegaR-1$ bipartite graphs,
where bipartite graph~$r$, $G_r$, $1 \leq r \leq \omegaR -1$, has two sides. The first side represents all the
servers of Type T$_r$ in which $x_i$ is neither a singleton nor in a sum with other parts.
The second side represents all the servers of Type T$_{r+1}$ in which $x_i$ participates in a sum with other parts.
There is an edge between vertex $v$ of the first side and vertex $u$ of the second side if the $t-1$ singleton parts
in $v$, and the $(r-1)t+1$ parts of the sum in the last cell of $v$ are the $rt$ parts in the sum of the
last cell of $u$, excluding $x_i$. We choose the constants~$\eta_r$ so that these bipartite graphs will all be all regular. Edges in a perfect matching of these graphs correspond to pairs of servers that can together recover $x_i$.

We start with a general solution for $\omegaR=3$ to show that this method is much better than
the previous ones. For $\omegaR=3$ there are three types of servers T$_1$, T$_2$, and~T$_3$.

In Type T$_1$, each server has $t$ singletons. There are $\binom{3t-1}{t-1}$ combinations
in which $x_i$ is a singleton and $\binom{3t-1}{t}$ combinations in which $x_i$ is not
a singleton. Each combination will appear in $\eta_1=\binom{2t-1}{t-1}$ servers of Type~T$_1$.

In Type T$_2$, each server has $t-1$ singletons and one cell with a sum
of $t+1$ parts. There are $\binom{3t-1}{t-2} \binom{2t+1}{t}$ combinations in which $x_i$ is a singleton,
$\binom{3t-1}{t-1} \binom{2t}{t}$ combinations in which $x_i$ is in a sum of $t+1$ parts, and
$\binom{3t-1}{t-1} \binom{2t}{t-1}$ combinations in which $x_i$ is neither a singleton nor in
a sum of $t+1$ parts. Each combination will appear in exactly one server of Type T$_2$, so $\eta_2=1$.

In Type T$_3$, each server has $t-1$ singletons and one cell with a sum
of $2t+1$ parts. Hence, each part appears in each server either as a singleton or in a sum
of $2t+1$ parts. There are $\binom{3t-1}{t-2}$ combinations in which $x_i$ is a singleton,
and $\binom{3t-1}{t-1}$ combinations in which $x_i$ is in a sum of $2t+1$ parts.
Each combination will appear in $\eta_3=8 \binom{2t}{t-1}$ servers of Type T$_3$.

Now, we can form the two bipartite graphs and apply Corollary~\ref{cor:Hall} to find the
pairs from which $x_i$ can be recovered. We may calculate that the rate of the code is $\frac{16 t^2 +7t +1}{24t^2 +15t +3}$,
which is much better than the rate of $\frac{4t+1}{6t+3}$ implied by Theorem~\ref{thm:s=integer}.
Hence, we have

\begin{theorem}
\label{thm:s=3}
$$
g(3,t) \geq \frac{16 t^2 +7t +1}{24t^2 +15t +3}~.
$$
\end{theorem}

The rate of the construction for each pair $(\omegaR,t)$, which is a lower bound on
$g(\omegaR,t)$, is given in the next theorem.
\begin{theorem}
For any integers $\omegaR$ and $t$ greater than one, the rate of the code by the construction
is $\frac{\beta + \gamma}{\beta + 2 \gamma}$, where
$$
\beta = \prod_{\ell =1}^{\omegaR-1} (\ell t +1) + (t-1)\sum_{r=2}^\omegaR \frac{(\omegaR-1)!}{(\omegaR-r)!} t^{r-2} \prod_{\ell =r}^{\omegaR-1} (\ell t+1)~,
$$
$$
\gamma = \sum_{r=1}^{\omegaR-1} \frac{(\omegaR -1)!}{(\omegaR-1-r)!} t^{r-1} \prod_{\ell =r}^{\omegaR-1} (\ell t+1)~.
$$
Moreover, when $t \rightarrow \infty$ the rate meets the upper bound of Theorem~\ref{thm:upper_bound}, i.e. $(\omegaR +1)/(2 \omegaR )$.
\end{theorem}

A careful analysis shows that the rate of this construction is larger
from the rates of the previous constructions when $\omegaR > 2$ (see~\cite{BlEt16}).

If $\omegaR$ is not an integer, then the construction is very similar. We note that there is some flexibility
in choosing the number of parts in each type (there is no such flexibility when~$\omegaR$ is an integer).
But we have to use the same types of servers as in the case
when~$\omegaR$ is an integer, except for the last type.
For example, consider the case when $t=3$ and $\omegaR=7/3$, so $p=7$.
There are three types of servers:

In Type T$_1$, each server has 3 singletons. There are 15 combinations
in which $x_i$ is a singleton and 20 combinations in which $x_i$ is not
a singleton. Each combination will appear in three servers of Type T$_1$, so~$\eta_1=3$.

In Type T$_2$, each server has 2 singletons and one cell with a sum
of four parts. There are 30 combinations in which $x_i$ is a singleton,
60 combinations in which $x_i$ is in a sum of four parts, and
15 combinations in which $x_i$ is neither a singleton nor in
a sum of four parts. Each combination will appear in exactly one server of Type~T$_2$, so $\eta_2=1$.

In Type T$_3$, each server has 3 singletons and one cell with a sum
of seven parts. Hence, each part appears in each server either as a singleton or in a sum
of seven parts. There are 6 combinations in which $x_i$ is a singleton,
and 15 combinations in which $x_i$ is in a sum of seven parts.
Each combination will appear in exactly one server, so $\eta_3=1$.

Now, the two bipartite graphs are formed and Corollary~\ref{cor:Hall} is applied to find the
pairs from which $x_i$ can be recovered. The rate of the resulting code is $\frac{52}{77}$
which is better than the $\frac{23}{35}$ rate implied by Theorem~\ref{thm:s=rational}. The rates
for other parameters are also better and a general rate for $w=7/3$ is given by:
\begin{theorem}
\label{thm:s=7_3}
$$
g(7/3,3t) \geq \frac{160 t^2 +45 t +3}{224t^2 +81t +7}~.
$$
\end{theorem}
%%%%%%%%%%%%%%%%%%%%%%%%%%%%%%%%%%
%%%
%%%       CONCLUSION
%%%
%%%%%%%%%%%%%%%%%%%%%%%%%%%%%%%%%%%%

%\section{Conclusions and Open Problems}
%\label{sec:conclude}
%
%We have constructed $k$-PIR array codes with good PIR rate for all possible pairs $(s,t)$,
%where a database is divided into $st$ parts, and each server stores $t$ linear
%combinations of parts in its cells. We have also proved upper bounds on
%the PIR rate of a PIR array code. These results are strong enough to determine
%the best rate of a PIR array code when $1<s\leq 2$ for all sensible choices
%of $t$. Moreover, the results determine the asymptotic value for the PIR rate
%for all rational values of $s$, when $t$ is allowed to tend to infinity.
%The research on PIR array codes is far from being complete.
%Some problems, in the direction taken in this paper are outlined in~\cite{BlEt16}.
\vspace{-0.07cm}

\section*{Acknowledgment}
\vspace{-0.07cm}
This work was supported in part by the EPSRC Grant EP/N022114/1.
\vspace{-0.14cm}
%\bibliography{allbib,extra}

\begin{thebibliography}{10}
\providecommand{\url}[1]{#1}
\csname url@rmstyle\endcsname
\providecommand{\newblock}{\relax}
\providecommand{\bibinfo}[2]{#2}
\providecommand\BIBentrySTDinterwordspacing{\spaceskip=0pt\relax}
\providecommand\BIBentryALTinterwordstretchfactor{4}
\providecommand\BIBentryALTinterwordspacing{\spaceskip=\fontdimen2\font plus
\BIBentryALTinterwordstretchfactor\fontdimen3\font minus
  \fontdimen4\font\relax}
\providecommand\BIBforeignlanguage[2]{{%
\expandafter\ifx\csname l@#1\endcsname\relax
\typeout{** WARNING: IEEEtran.bst: No hyphenation pattern has been}%
\typeout{** loaded for the language `#1'. Using the pattern for}%
\typeout{** the default language instead.}%
\else
\language=\csname l@#1\endcsname
\fi
#2}}



\bibitem{ALS14}
{D. Augot, F. Levy-Dit-Vahel, and A. Shikfa,}
``A storage-efficient and robust private information retrieval scheme allowing few servers,''
in {\em Cryptology and Network Security},  pp.\,222--239, Springer 2014.

\bibitem{BlEt16}
{S. R. Blackburn and T. Etzion}, ``PIR array codes with optimal PIR rate,''
\emph{arxiv.org/abs/1607.00235}, July 2016.

\bibitem{BEP16}
{S. R. Blackburn, T. Etzion, and M. B. Paterson}, ``PIR schemes with small download complexity and low storage requirements,''
\emph{arxiv.org/abs/1609.07027}, September 2016.

%\bibitem{Knu98}
%    D. E. Knuth,
%    {\em The Art of Computer Programming, Volume 3: Sorting and Searching},
%    Reading, MA: Addiaon-Wesley, 1998.

\bibitem{CHY14}
{T. H. Chan, S. Ho, and H. Yamamoto}, ``Private information retrieval for coded storage,''
\emph{arxiv.org/abs/1410.5489}, October 2014.

\bibitem{CGKS98}
{B. Chor, O. Goldreich, E. Kushilevitz, and M. Sudan}, ``Private information retrieval,''
\emph{Journal ACM}, pp.\,965--981, 1998.

%\bibitem{FaRa15}
%{G. Fanti and K. Ramchandran}, ``Efficient private information retrieval over unsynchronized databases,''
%\emph{IEEE J. on Selected Topics in Signal Processing}, vol.\,9, pp.\,1229--1239, 2015.

\bibitem{FVY15}
{A. Fazeli, A. Vardy, and E. Yaakobi}, ``Coded for distributed PIR with low storage overhead,''
\emph{Proc. IEEE International Symposium on Information Theory}, pp.\,2852--2856, Hong Kong, June 2015.

\bibitem{FVY15a}
{A. Fazeli, A. Vardy, and E. Yaakobi}, ``Private information retrieval without storage overhead: coding instead of replication,''
\emph{arxiv.org/abs/1505.0624}, May 2015.

\bibitem{Hal35}
{P. Hall}, ``On representatives of subsets,''
\emph{Journal of London Mathematical Society}, vol.\,10, pp.\,26--30, 1935.

\bibitem{KuOs97}
{E. Kushilevitz and R. Ostrovsky}, ``Replication is not needed: Single database, computationally-private information retrieval,''
\emph{Proc. 38-th IEEE Symp. Foundations Computer Science (FOCS)}, pp.\,364--373, 1997.

\bibitem{SRR14}
{N. Shah, K. Rashmi, and K. Ramchandran}, ``One extra bit of download ensures perfectly private information retrieval,''
\emph{Proc. IEEE International Symposium on Information Theory}, pp.\,856--860, Honolulu, Hawaii, June 2014.

\bibitem{RaVa16}
{S. Rao and A. Vardy}, ``Lower bound on the redundancy of PIR codes,''
\emph{arxiv.org/abs/1605.01869}, May 2016.

\bibitem{TaElR16}
{R. Tajeddine and S. El Rouayheb}, ``Private information retrieval from MDS coded data in distributed storage systems,''
\emph{Proc. IEEE International Symposium on Information Theory}, pp.\,1411--1415, Barcelona, Spain, July 2016.

\bibitem{ZWWG16}
{Y. Zhang, X. Wang, N. Wei, and G. Ge}, ``On private information retrieval array codes,''
\emph{arxiv.org/abs/1609.09167}, September 2016.

\end{thebibliography}

\end{document}